\documentclass[11 pt]{amsart}

\usepackage[latin1]{inputenc}
\usepackage[T1]{fontenc}

\usepackage{amsfonts,amssymb,amsmath,amsthm}
\usepackage[headinclude,DIV13]{typearea}
\usepackage{hyperref}
\usepackage{geometry}
\usepackage{graphicx, psfrag}
\usepackage{array}

 \usepackage[dvipsnames]{color}
 \usepackage{setspace}

\newtheorem{lemma}{Lemma}[section]

\newtheorem{theorem}{Theorem}[section]

\renewcommand{\P}{\mathbb{P}}
\newcommand{\R}{\mathbb{R}}

\newcommand{\E}{\mathbb{E}}

\newcommand{\N}{\mathbb{N}}

\newcommand{\CA}{\mathcal{A}}
\newcommand{\tb}[1]{{\mathbf{#1}}}

\newcommand{\eps}{\varepsilon}

\numberwithin{equation}{section}

\begin{document}

\title[Asymmetrical mating preferences]{A stochastic model for reproductive isolation under asymmetrical mating preferences}

\author{H\'el\`ene Leman}
\address{CIMAT, De Jalisco S-N, Valenciana, 36240 Guanajuato, Gto., Mexico}
\email{helene.leman@cimat.mx}

\maketitle

\begin{abstract}
More and more evidence shows that mating preference is a mechanism that may lead to a reproductive isolation event. In this paper, a haploid population living on two patches linked by migration is considered. Individuals are ecologically and demographically neutral on the space and differ only on a trait, $a$ or $A$, affecting both mating success and migration rate. The special feature of this paper is to assume that the strengths of the mating preference and the migration depend on the trait carried. Indeed, patterns of mating preferences are generally asymmetrical between the subspecies of a population.
I prove that mating preference interacting with frequency-dependent migration behavior can lead to a reproductive isolation. Then, I describe the time before reproductive isolation occurs. To reach this result, I use an original method to study the limiting dynamical system, analyzing first the system without migration and adding migration with a perturbation method. Finally, I study how the time before reproductive isolation is influenced by the parameters of migration and of mating preferences, highlighting that large migration rates tend to favor types with weak mating preferences. 
\end{abstract}

\medskip \noindent\emph{Keywords:} mating preference, asymmetrical preference, birth-death stochastic model, dynamical system, long-time behaviour, perturbation method.

\medskip
\noindent\emph{AMS subject classification:} 92D40, 37N25, 60J27.

\maketitle

\section{Introduction}

Understanding the mechanisms of speciation and reproductive isolation is a major issue in evolutionary biology. 
There are now strong evidence that sexual preferences and speciation are tied~\cite{lande1981models,boughman2001divergent}. The role of 'magic' or 'multiple effect' trait, which associates both adaptation to an ecological niche and a mate preference, has first been studied deeply. It has been shown that such traits may lead to speciation using direct experimental evidence \cite{merrill2012disruptive} or theoretical works \cite{lande1988ecological,van1998sympatric}. 
Then, studies focused on the particular role of mating preference during a speciation event \cite{gavrilets2014review}, highlighting that (i) it may impede reproductive isolation~\cite{Servedio2010,servedio2014counterintuitive,servedio2015effects}, or, (ii) it may promote reproductive isolation. This promoting role may be secondary or primary. For example, the initial divergence in traits may be the result of natural selection in order to decrease hybridization and then be subjected to mating preference~\cite{panhuis2001sexual}, producing speciation by reinforcement~\cite{gregorius1989characterization}. Other studies illustrate the direct and promoting role of assortative mating,
using numerical simulations \cite{kondrashov1998origin,m2012sexual,smadi2017assortative}, or theoretical results \cite{rudnicki2015model,coron2016stochastic}.

The studies mentioned above focus on a symmetrical sexual preference, assuming that all individuals express the same sexual preference. Numerous observations and studies though do not support this assumption and describe examples of species that express different patterns of preference (See~\cite{panhuis2001sexual} for examples).~\cite{smadja2005asymmetrical} describe such an example between two subspecies of the house mouse. The subspecies \textit{Mus musculus musculus} is characterized by a stronger assortative preference than the subspecies \textit{Mus musculus domesticus} \cite{smadja_catalan_2004}. 
A mechanism for subspecies recognition mediated by urinary signals occurs between these two taxa and seems to maintain a reproductive isolation. Another example comes from \textit{Drosophila melanogaster} populations where a strong sexual isolation with an asymmetrical pattern of sexual preference was observed \cite{wu1995sexual,hollocher1997incipient}. The Zimbabwe female lines of \textit{Drosophila melanogaster} have a nearly exclusive preference for males from the same locality over the males from other regions or continents; the reciprocal mating is also reduced but to a lesser degree.

Hence, in this paper, I focused on the cases where mating preference promotes sexual selection, and I was interested in two main problematic: (i) studying the influence of an asymmetric mating preference pattern on speciation mechanisms, and (ii) understanding the effects of migration on mating preference advantages.
To do so, I aimed to generalize the model of~\cite{coron2016stochastic} to account for asymmetrical sexual preferences. A haploid population divided in two demes but connected by migration is considered. Following the seminal papers \cite{bolker1997using,dieckmann2000relaxation,fournier2004microscopic}, I
used a stochastic individual-based model with competition and varying population size. Individuals are assumed not to express any local adaptation. Their parameters do not depend on their location. Individuals, however, are characterized by a mating trait, encoded by a bi-allelic locus, and which has two consequences: (i) individuals of the same type have a higher probability to mate and give an offspring, and (ii) the migration rate of an individual increases with the proportion of individuals caring the other trait in its deme. Finally, the two alleles may not have identical effects, in the sense that strengths of mating preferences and of migration depend on the allele carried by the individual.

Using convergence to the large population limit, I first connected the microscopic model to a macroscopic and deterministic model. Then, studying  
both models, I established the main result of the paper, which ensures that the mechanism of mating preference combined with a negative type-dependent migration is sufficient to entail reproductive isolation and which gives the time needed before reproductive isolation. Here, unlike \cite{coron2016stochastic}, the time is written with both mating preference parameters and both migration parameters, related to both alleles. I finally conducted an extensive study on the influence of migration and preference parameters on this time showing that large migration rates can favor types with a weak mating preference. 
The proof of the main result is based on a fine analysis of the deterministic limiting model. In particular, global results on the dynamical system are established such that the dynamics of almost all trajectories can be predicted. To do so, I developed an original method based on a perturbation theory of the migration parameters, which fully differs from the method used by~\cite{coron2016stochastic}. The asset of this method is that it can easily be adapted to other dynamical systems.

The paper is organized as follows. In Section~\ref{sec_model}, the stochastic model is introduced and motivated from a biological perspective. Section~\ref{sec_result} presents the results of the paper. In particular, the main results on the deterministic limiting model and on the stochastic process are stated in Section~\ref{subsec_time}. Section~\ref{subsec_withoutmig} presents the main result in the case without migration between both patches. In Section~\ref{subsec_num}, the influence of migration on the time before reproductive isolation is analyzed. Finally, Section~\ref{subsec_perturbation} establishes the proof of the key result using perturbation theory. Proofs of the case without migration will be found in Appendix~\ref{subsec_sansp}. Proofs of probabilistic parts of the main result will be found in Appendix~\ref{subsec_probaasym}.

\section{Model}
\label{sec_model}
The population is divided into two patches. The individuals are haploid and characterized by a diallelic locus ($a$ or $A$) and a position ($1$ or $2$ depending on the patch in which they are). The set  $\mathcal{E}:=\{(\alpha,i), \alpha \in \{a,A\}, i\in\{1,2\}\}$ is used to characterize the individuals. The population dynamics follows a multi-type birth and death process with competition in continuous time.  In other words, the dynamics follows a Markov jump process in space $\N^\mathcal{E}$, where the rates are described below.
At any time $t$, the population is represented by the following vector of dimension $4$ in $\N^\mathcal{E}$ : 
\begin{equation*}
 {\bf{N}}^K(t)=(N^K_{A,1}(t),N^K_{a,1}(t),N^K_{A,2}(t),N^K_{a,2}(t)) \in \N^\mathcal{E}
\end{equation*}
where $N^K_{\alpha,i}(t)$ denotes the number of individuals with genotype $\alpha$ in the deme $i$ at time $t$.
$K$ is an integer parameter associating to the concept of carrying capacity and accounting for the quantity of available resources or space (see also \cite{coron2016stochastic} for more details). Consequently, it is a scaling parameter for the size of the community. It is assumed to give the order of magnitude of the initial population, in the sense that the initial number of individuals divided by $K$ converges (in probability) when $K$ goes to infinity. The competition for resources is also scaled with $1/K$, as presented below.

 In what follows, if $\alpha$ denotes one of the alleles, notation $\bar \alpha$ denotes the other allele and if $i$ denotes one of the demes, $\bar i$ denotes the other one. \\

The birth, death and migration rates of each individual are now described.

At a rate $B>0$, a given individual with trait $\alpha\in \{a,A\}$ encounters uniformly at random another individual of its deme. Then it mates with the latter and transmits its trait with probability $b\beta_\alpha /B \leq 1$ if the other individual carries also the trait $\alpha$, and with probability $b/B\leq 1$ if the other individual carries the trait $\bar\alpha$. That is to say, after encountering, two individuals that carry the same trait $\alpha$ have a probability $\beta_\alpha$-times larger to mate and give birth to a viable offspring than two mating individuals with different traits. Hence, the current state of the population is denoted by ${\bf{N}}^K\in \N^{\mathcal{E}}$, the total birth rate of $\alpha$-individuals in patch $i$ is 
\begin{equation}
\label{eq_birth}
 bN^K_{\alpha,i} \frac{\beta_\alpha N^K_{\alpha,i}+N^K_{\bar\alpha,i}}{N^K_{\alpha,i}+N^K_{\bar\alpha,i}}.
\end{equation} 
Note that two parameters, $\beta_a$ and $\beta_A$, are used to model the sexual preference depending on the trait carried by the individual. The limiting case where $\beta_A=\beta_a$ was studied by~\cite{coron2016stochastic}. Here, I was interested in the case where $\beta_a \neq \beta_A$ although the result of the limiting case can be rediscovered with our calculation.
As presented in \cite{coron2016stochastic}, Formula~\eqref{eq_birth} models an assortative mating by phenotypic matching or recognition alleles \cite{blaustein1983kin,JonesRatterman2009}. 
Note that, in the present model, preference modifies the rate of mating and not only the distribution of genotypes, unlike what is usually assumed in classical generational models  \cite{odonald1960assortive,lande1981models,Kirkpatrick1982,gavrilets2004fitness,BurgerSchneider2006,Servedio2010}.
The present model can be compared with these classical ones by computing the probabilities that an individual of trait $\alpha$ in the deme $i$ gives birth after encountering an individual of the same trait (resp. of the opposite trait) conditionally on the fact that this individual gives birth at time $t$, and we find
\begin{equation*}
 \frac{\beta_\alpha N^K_{\alpha,i}}{\beta_\alpha N^K_{\alpha,i}+N^K_{\bar\alpha,i}} \qquad \left(\text{resp.} \quad \frac{ N^K_{\bar\alpha,i}}{\beta_\alpha N^K_{\alpha,i}+N^K_{\bar\alpha,i}} \right).
\end{equation*}
Note that these terms correspond to the ones presented in the supplementary material of~\cite{Servedio2010}, or in~\cite{gavrilets1998evolution}. A extended discussion between these two types of models can be found in Section 2 of \cite{coron2016stochastic}.

The death rate of a given individual is composed of a natural death rate and a competition death rate. Individuals compete for resources or space against all individuals of its own deme. The competitive death rate of each individual is thus proportional to the total population size of its deme. Finally, the total death rate of $\alpha$-individuals in patch $i$ is 
\begin{equation}
 \left(d+\frac{c}{K}\left(N^K_{\alpha,i}+N^K_{\bar\alpha,i}\right)\right)N^K_{\alpha,i},
\end{equation}
where $d$ models the natural death and $c$ models the competition for resources. As presented previously, $K$ is the scaling parameter that scales the amount for resources. Hence, the larger $K$ is, the smaller the strength of competition between two individuals, $c/K$, is.

Finally, individuals can migrate from one patch to the other one. Following \cite{PayneKrakauer1997,coron2016stochastic,smadi2017assortative}, I use a density-dependent migration rate in such a way that individuals are more prone to move if they do not find a suitable mate. This hypothesis is relevant for all organisms with
active mate searching \cite{uy2001complex,jugovic2017movement}. The migration term of an individual is proportional to the proportion of individuals carrying the other allele in its deme, and to a parameter $m_\alpha$ which depends on the trait of the individual. Hence, the alleles code for the strength of the mating preference and simultaneously, the speed of migration. The total migration rate of $\alpha$-individuals from patch $1$ to patch $2$ finally is 
\begin{equation}
m_\alpha \left( \frac{N^K_{\bar \alpha,1}}{N^K_{\alpha,1}+N^K_{\bar\alpha,1}} \right) N^K_{\alpha,1}.
\end{equation} 
Note that the migration rate does not depend on the other deme composition.\\

In what follows, the following statements on the parameters are assumed:
$$
\beta_A>1, \;\; \beta_a>1, \;\; b>d>0, \;\; c>0, \;\; m_A\geq 0, \;\; m_a\geq 0.
$$

\section{Results}
\label{sec_result}

\subsection{Time needed before reproductive isolation}
\label{subsec_time}
In this section, I present the main result of the paper that gives the time needed for the process ${\bf{N}^K}$ to reach reproductive isolation. This time is given with respect to $K$, the carrying capacity of the process.\\

To this aim, let us first give the process average behavior using convergence to the large limit population. Precisely, Lemma~\ref{lemma_1} below ensures that the sequence of re-scaled processes
$$
({\bf{Z}}^K(t))_{t\geq 0} = \left( \frac{{\bf{N}}^K(t)}{K} \right)_{t\geq 0}
$$
converges when $K$ goes to infinity to 
\begin{equation}\label{systdetasym}
\left\{\begin{aligned} \frac{d}{dt}z_{A,1}(t)&=
z_{A,1}\Bigl[ b\frac{\beta_A z_{A,1}+z_{a,1}}{z_{A,1}+z_{a,1}}-d-c(z_{A,1}+z_{a,1})-m_A\frac{z_{a,1}}{z_{A,1}+z_{a,1}}\Bigr]+m_A\frac{z_{A,2}z_{a,2}}{z_{A,2}+z_{a,2}}\\
\frac{d}{dt}z_{a,1}(t)&=
z_{a,1}\Bigl[ b\frac{\beta_a z_{a,1}+z_{A,1}}{z_{A,1}+z_{a,1}}-d-c(z_{A,1}+z_{a,1})-m_a\frac{z_{A,1}}{z_{A,1}+z_{a,1}}\Bigr]+m_a\frac{z_{A,2}z_{a,2}}{z_{A,2}+z_{a,2}}\\
\frac{d}{dt}z_{A,2}(t)&=
z_{A,2}\Bigl[ b\frac{\beta_A z_{A,2}+z_{a,2}}{z_{A,2}+z_{a,2}}-d-c(z_{A,2}+z_{a,2})-m_A\frac{z_{a,2}}{z_{A,2}+z_{a,2}}\Bigr]+m_A\frac{z_{A,1}z_{a,1}}{z_{A,1}+z_{a,1}}\\
\frac{d}{dt}z_{a,2}(t)&=
z_{a,2}\Bigl[ b\frac{\beta_a z_{a,2}+z_{A,2}}{z_{A,2}+z_{a,2}}-d-c(z_{A,2}+z_{a,2})-m_a\frac{z_{A,2}}{z_{A,2}+z_{a,2}}\Bigr]+m_a\frac{z_{A,1}z_{a,1}}{z_{A,1}+z_{a,1}}.\end{aligned}\right.
\end{equation}
\begin{lemma}
\label{lemma_1}
Assume that the sequence $({\bf{Z}}^K(0))_{ K\geq 0}$ converges in probability to the deterministic vector ${\bf{z}}_0 \in \R^{\mathcal{E}}$. Then, for any $T\geq 0$, 
\begin{equation}\label{eq_conv}
\underset{K \to \infty}{\lim}\  \sup_{s\leq T}\ \|{\bf{Z}}^K(s)-{\bf{z}}^{(\bf{z}_0)}(s)  \|=0 \quad \text{in probability},
\end{equation}
where $\|. \|$ denotes the $L^\infty$-Norm on $\R^\mathcal{E}$ and $({\bf{z}}^{(\bf{z}_0)}(t) )_{t\geq 0}$ denotes the solution of~\eqref{systdetasym} with initial condition ${\bf{z}}_0\in \R^\mathcal{E}$
\end{lemma}
This result can be deduced from a direct application of Theorem 2.1 p. 456 by~\cite{EK}.

A direct computation implies that the following four points are stable equilibria of the system~\eqref{systdetasym}:
\begin{itemize}
\item equilibria with fixation of an allele (where only an allele is maintained in both patches)
\begin{equation}
 \label{eq_stableasym1}
(\zeta_A,0,\zeta_A,0),\;   (0,\zeta_a,0,\zeta_a),
\end{equation}
\item equilibria with maintenance of each allele in a different patch
\begin{equation}
 \label{eq_stableasym2}
(\zeta_A,0,0,\zeta_a),\;  (0,\zeta_a,\zeta_A,0),
\end{equation}
\end{itemize}
with $\zeta_\alpha:=\frac{b\beta_\alpha-d}{c}$, $\alpha\in \{A,a\}$.
These four equilibria describe states of reproductive isolation: Once reaching one of these equilibria, migration rates equals zero and no individual travels anymore. 
More specifically, observe that equilibria \eqref{eq_stableasym2} are of particular interest to our problematic. Indeed, once reaching one of these equilibria, even if a small basal migration (i.e. constant migration) is added, the mating preferences and the frequency-dependent migration terms will prevent the populations of both demes to mix again, leading to \textit{migration-selection balance} \cite{karlin1972polymorphisms} but where selection is due to sexual selection and not to natural selection. Precisely, if an $A$-individual travels because of basal migration from patch $1$ to patch $2$, which is filled with $a$-individuals, its probability to reproduce will be significantly reduced in patch $2$ and its migration rate to come back will be so high that it is quite unlikely that its offspring establish in patch $2$. This reasoning, however, fails with equilibria \eqref{eq_stableasym1}.

Our aim is then to understand the long-time behavior of trajectories of the dynamical system and more specifically to detail the set of initial states that lead to one of these equilibria, which corresponds exactly to the basin of attraction of this equilibrium. With this aim in mind, let us define the weighted sums
\begin{eqnarray*}
&\Sigma_i:= (\beta_A-1)z_{A,i}+(\beta_a-1)z_{a,i}, \text{ for } i=1,2,\\
&\Sigma:= \Sigma_1+\Sigma_2 = (\beta_A-1)(z_{A,1}+z_{A,2})+(\beta_a-1)(z_{a,1}+z_{a,2}),
\end{eqnarray*}
and the compact set
\begin{multline}
 \label{def_setS}
\mathcal{S}:=\left\{ {\bf z}\in \R^{\mathcal{E}}, \Sigma_i \geq \frac{(\beta_{\min} -1)(b-d)}{4c}, \text{ for } i=1,2,
 \quad \text{and}\quad
\Sigma \leq \frac{4(\beta_{\max}-1)(b\beta_{\max} -d)}{c}  \right\},
\end{multline}
where $\beta_{\min}=\min(\beta_a,\beta_A)$ and $\beta_{\max}=\max(\beta_a,\beta_A)$.
Next Lemma ensures that we can focus on trajectories starting from $\mathcal{S}$
since any trajectory reaches it in finite time.
\begin{lemma}
\label{lemma_invsetasym}
 Assume that 
\begin{equation}
\label{ass_pasym}
m_A(\beta_A-1)+m_a(\beta_a-1)\leq 2(b-d)(\beta_{\min}-1).
\end{equation}
$\mathcal{S}$ is a positively invariant set for the dynamical system~\eqref{systdetasym}, in the sense that any trajectories starting from this set will stay in it when $t$ grows to $+\infty$. Moreover, any trajectory solution of~\eqref{systdetasym} hits $\mathcal{S}$ after a finite time.
\end{lemma}
The aim is thus to study trajectories inside the compact set $\mathcal{S}$.
\begin{theorem}
\label{theo_sysasym}
\begin{enumerate}
\item Assume that $m_A=0$ if and only if $m_a= 0$. There exists $m_0>0$ such that for all $m_A\leq m_0$ and $m_a\leq m_0$, there exist four open subsets $(\mathcal{D}^{\alpha,\alpha'}_{m_A,m_a})_{\alpha,\alpha'\in \CA}$ of $\mathcal{S}$ that are the basins of attraction in $\mathcal{S}$ of the four equilibria~\eqref{eq_stableasym1} and \eqref{eq_stableasym2} of the system~\eqref{systdetasym}, and such that the closure of $\cup_{\alpha,\alpha'\in \CA} \mathcal{D}^{\alpha,\alpha'}_{m_A,m_a}$ is equal to $\mathcal{S}$.\\
 \item In the case $m_A=m_a=0$, the basins of attraction are exactly
 $$
\mathcal{D}^{\alpha,\alpha'}_{0,0} = \left\{ {\bf z}\in \R^{\mathcal{E}}, (\beta_\alpha-1)z_{\alpha,1} > (\beta_{\bar\alpha}-1)z_{\bar\alpha,1} \text{ and } (\beta_{\alpha'}-1)z_{\alpha',2} > (\beta_{\bar\alpha'}-1)z_{\bar\alpha',2}  \right\} \cap \mathcal{S}.
$$
\end{enumerate}
\end{theorem}
Theorem~\ref{theo_sysasym} ensures that any trajectory starting from $\mathcal{S}$, except from an empty interior set, reaches one of the steady states~\eqref{eq_stableasym1} or \eqref{eq_stableasym2}.
In particular, coexistence of both alleles in a single deme does never occur. 
Hence, assortative mating combined with negative type-dependent migration entails reproductive isolation. The assumption on the migration rate is essential to obtain this result. Different results are deduced in models with frequency-independent migration~\cite{servedio2014counterintuitive,smadi2017assortative}. In particular, reproductive isolation may be prevented. 
\cite{smadi2017assortative} study a similar model as the one used here where individuals are diploid. A mechanism of mating preference interacting with frequency-dependent migration is studied. In Section 3.4 of this paper, the frequency-dependent migration term is replaced by a constant migration term. Then, polymorphic equilibria with both alleles in demes can only be observed if the migration rate is sufficiently large. This highlights that, although using other kind of migration prevents reproductive isolation, the mechanism that would prevent reproductive isolation is the migration and not the assortative mating, in their case as in the one presented here.

 Theorem~\ref{theo_sysasym} is, furthermore, a key result to deduce the next theorem, which gives the time before reproductive isolation.
It can be compared to Theorem 2 of \cite{coron2016stochastic} which gives same results in the symmetrical case ($\beta_A=\beta_a$ and $m_A=m_a$). 
In the latter, the equilibrium reached is given by the alleles that make up the majority initially in each patch. In our case, the dynamics is more involved. 
Without migration, the equilibrium reached depends on the initial number of individuals of each type and of the mating preference strengths.
Then, when $m_A$ and $m_a$ are small, the basin of attraction $\mathcal{D}^{\alpha,\alpha'}_{m_A,m_a}$ is a continuous deformation of $\mathcal{D}^{\alpha,\alpha'}_{0,0}$.  I drew such an example in Section~\ref{subsec_num}. Note that no basin of attraction is empty, since the four equilibria are stable equilibria.
 
The asymmetrical sexual preferences make the long-time behavior more involved than in the symmetrical case and proofs here use completely different mathematical techniques. 
I used perturbation theory to deduce Theorem~\ref{theo_sysasym}~: the system is first studied in the particular case where $m_A=m_a=0$, then one makes $m_A$ and $m_a$ grow up to deduce the result for positive migration rates.
Unfortunately, I was not able to give an explicit formulation for the sets $\mathcal{D}^{\alpha,\alpha'}_{m_A,m_a}$ unlike in the symmetrical case.\\

Let us now state the main result. It describes the random time $T^K_{\mathcal{B}_{A,a,\eps}}$  that is the first time when the population process ${\bf{N}}^K$ reaches the set
$$ \mathcal{B}_{A,a,\eps}:= [(\zeta_A-\eps)K,(\zeta_A+\eps)K] \times \{0\} \times \{0\} \times [(\zeta_{a}-\eps)K,(\zeta_{a}+\eps)K], $$
with $\eps>0$ and when $K$ is large. In other words, it is the random time before (1) all $a$-individuals in patch $1$ and all $A$-individuals in patch $2$ get extinct, and (2) the population size in patch $1$ is approximately $K\zeta_A$ and the one in patch $2$ is approximately $K\zeta_a$. In the light of the previous discussion about equilibrium $(\zeta_A,0,0,\zeta_a)$, it thus corresponds to the time before reproductive isolation occurs.
\begin{theorem}
\label{theo_asym}
Assume that Assumptions of Theorem~\ref{theo_sysasym} holds and that $m_A\leq m_0$ and $m_a\leq m_0$.\\
  Let $\varepsilon_0>0$ and assume also that ${\bf{Z}}^K(0)={\bf{N}}^K(0)/K$ converges in probability to a deterministic vector ${\bf{z}^0}\in\mathcal{D}^{A,a}_{m_A,m_a}$ such that $(z^0_{a,1},z^0_{A,2})\neq (0,0)$.
Then there exist $C_0>0$, $M>0$, and $V>0$ depending only on $(M,\eps_0)$ such that, for any $\eps\leq \eps_0$,
 \begin{equation}
  \label{eq_probatheo}
  \lim_{K \to \infty}\P \left(\left\vert \frac{T^K_{\mathcal{B}_{A,a,\eps}}}{\log K}- \frac{1}{\omega(A,a)} \right\vert  \leq  C_0\eps ,  \; 
 \tb{N}^K\left(T^K_{\mathcal{B}_{A,a,\eps}}+t\right)\in \mathcal{B}_{A,a,M\eps}; \; \forall t \leq e^{VK} \right)= 1,
 \end{equation}
 where for all $\alpha, \alpha' \in \CA$,
\begin{equation}
 \label{def_constAa}
\omega(\alpha,\alpha')=\frac{1}{2}\left[ b (\beta_{\alpha}-1+\beta_{\alpha'}-1)+m_\alpha+m_{\alpha'}-\sqrt{\Big(b(\beta_{\alpha}-\beta_{\alpha'})+(m_{\alpha'}-m_\alpha)\Big)^2+4m_{\alpha}m_{\alpha'}} \right].
\end{equation}
Similar results hold for the three other equilibria of~\eqref{eq_stableasym1} and \eqref{eq_stableasym2}.
\end{theorem}

Theorem~\ref{theo_asym} gives the first-order approximation of the time before reproductive isolation. The latter is proportional to $\log(K)$, which is short comparing to $K$, the order of magnitude of the population size. 
Comparatively, the time scale needed for the random genetic drift to entail the end of gene flow between two populations is of order $K$ in a large amount of models.
Hence, Theorem~\ref{theo_asym} implies that reproductive isolation due to mating preference is much shorter.
Note also that Theorem \ref{theo_asym} not only gives the time before reproductive isolation but also it ensures that once the equilibrium is reached, the population sizes of both patches stay around $K\zeta_\alpha$ during at least a long time of order $e^{KV}$.
Secondly, as $\omega(\alpha,\alpha)=b(\beta_{\alpha}-1)$, the time before reaching one of equilibria~\eqref{eq_stableasym1} does not depend on migration parameters unlike the time before reaching one of equilibria~\eqref{eq_stableasym2}. I studied more specifically the influence of migration parameters on this time in Section~\ref{subsec_num}.

The assumption on initial condition ($(z^0_{a,1},z^0_{A,2})\neq (0,0)$) is only needed to obtain the lower bound on the time $T^K_{\mathcal{B}_{A,a,\eps}}$ given in~\eqref{eq_probatheo}. Otherwise, this time would be faster.
Finally, note that, assuming $\beta_A = \beta_a$ and $m_A=m_a$, Theorem 3 of \cite{coron2016stochastic} is rediscovered.

\subsection{Study of the system without migration}
\label{subsec_withoutmig}
The proofs of Theorems \ref{theo_sysasym} and \ref{theo_asym} require a full understanding of the dynamics without migration. Hence before proceeding with the proofs, I present a complete study of the dynamical system when $m_A=m_a=0$. Since both patches evolve independently in this case, only the dynamics of patch 1 is studied and, for the sake of simplicity, the dependency on patches in notation is dropped. From~\eqref{systdetasym}, we find that
\begin{equation}
\label{eq_system1patch}
\left\{
\begin{aligned}
&\frac{d}{dt}z_{A}(t)=
z_{A}\Big[ b\frac{\beta_A z_{A}+z_{a}}{z_{A}+z_{a}}-d-c(z_{A}+z_{a})\Big]\\
&\frac{d}{dt}z_{a}(t)=
z_{a}\Big[ b\frac{\beta_a z_{a}+z_{A}}{z_{A}+z_{a}}-d-c(z_{A}+z_{a})\Big]\\
\end{aligned}
\right.
\end{equation}

The equilibria of the system will be written with the following quantities
\begin{equation*} 
\chi_\alpha:=(\beta_{\bar \alpha}-1)\chi, \quad \text{where} \quad  \chi:=\frac{b(\beta_a-1)(\beta_A-1)+(b-d)(\beta_A-1+\beta_a-1)}{c(\beta_A-1+\beta_a-1)^2},
\end{equation*}
and where $\bar\alpha$ is the complement of $\alpha\in \CA$.
A direct computation implies that there exist exactly four fixed points of the dynamical system~\eqref{eq_system1patch}:
\begin{equation*}
\begin{array}{rl}
&(0,0), \quad (\zeta_A,0), \quad (0,\zeta_a), \quad  \text{and} \quad (\chi_A,\chi_a).\\
\end{array}
\end{equation*}
These equilibria represent respectively the extinction of the population, the loss of allele $a$ or allele $A$, or the long-time coexistence of both alleles.

Let us now describe their stability and the long time behavior of any solution of~\eqref{eq_system1patch}.
\begin{lemma}
\label{lemma_sanspasym}
 \begin{itemize}
  \item $(\zeta_A,0)$ and $(0,\zeta_a)$ are two stable equilibria, $(0,0)$ is unstable and $(\chi_A,\chi_a)$ is a saddle point.
 \item The set
\begin{equation}
 \label{def_setDasym}
\mathcal{D}^A_0:=\left\{ (z_A,z_a)\in \R^2, (\beta_A-1)z_A > (\beta_a-1)z_a  \right\}
\end{equation}
is a positively invariant set under the dynamical system~\eqref{eq_system1patch}. Moreover, any solution starting from $\mathcal{D}^A_0$ converges to $(\zeta_A,0)$ when $t$ converges to $+\infty$.
 \item The set
\begin{equation*}
 \mathcal{D}^a_0:=\{(z_A,z_a)\in \R^2, (\beta_A-1)z_A < (\beta_a-1)z_a \},
\end{equation*}
is a positively invariant set under the dynamical system~\eqref{eq_system1patch}. Any solution starting from $\mathcal{D}^a_0$ converges to $(0,\zeta_a)$ when $t$ converges to $+\infty$.
 \item Finally, $\{(z_A,z_a)\in \R, (\beta_A-1)z_A=(\beta_a-1)z_a\}$ is also a positively invariant set and any solution starting from this set converges to $(\chi_A,\chi_a)$ when $t$ grows to $+\infty$.\\
 \end{itemize}
\end{lemma}
In other words, the system is bi-stable: All trajectories converge to $(\zeta_A,0)$ or $(0,\zeta_a)$, except the trajectories starting from a line (see Fig.~\ref{fig_phase}).
A direct consequence of this Lemma is that the basin of attraction $\mathcal{D}^{\alpha,\alpha'}_{0,0}$ are the ones described by Theorem~\ref{theo_sysasym}. 
\begin{figure}[h]
\begin{center}
\includegraphics[width=0.7\textwidth]{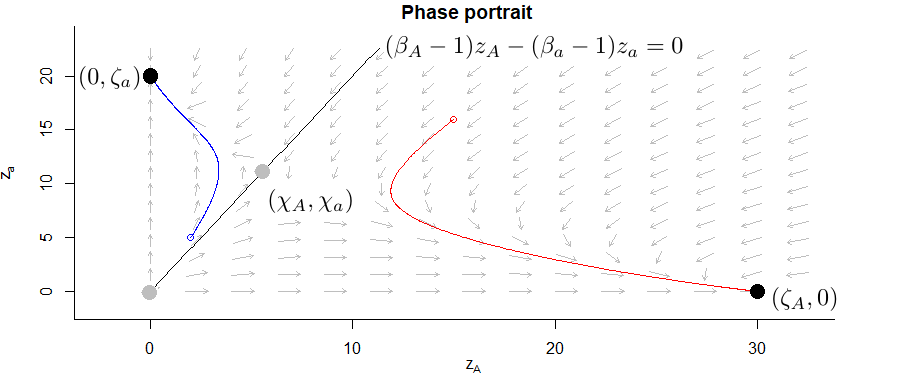}
\caption{\label{fig_phase} Example of phase portrait of a single patch dynamics. The black line is the set of initial conditions for which trajectories converge to the unstable equilibrium $(\chi_A,\chi_a)$. Above (resp. below) this line, trajectories converge to $(0,\zeta_a)$ (resp. $(\zeta_A,0)$). The red and the blue curves are examples of trajectories. The black and the gray points represent respectively the stable and the unstable equilibria.}
\end{center}
\end{figure}

\subsection{Parameters influence on the time before reproductive isolation}
\label{subsec_num}
In this section, the model under study is the initial one with two demes. I used functional studies and simulations to explore the influence of migration rates and mating preference parameters on the process. The simulations below were computed with the following demographic parameters:
\begin{equation}
\label{param_simu}
\beta_A=2, \quad \beta_a=1.5, \quad b=2, \quad d=1 \quad \text{and} \quad c=0.1,
\end{equation}
unless stated otherwise. For these parameters, 
$$\zeta_A =30 \qquad \text{and} \qquad \zeta_a=20.
$$

\noindent 
\textbf{Influence of parameters on the time before reproductive isolation}: Assume that the process starts from a state ${\bf z}^0\in \mathcal{D}^{A,a}_{m_A,m_a}$. Then, according to Theorems \ref{theo_sysasym} and \ref{theo_asym}, the trajectory will reach a neighborhood $\mathcal{B}_{A,a,\eps}$ of equilibrium $(\zeta_A,0,0,\zeta_a)$ after a time of magnitude $\log(K)\omega(A,a)^{-1}$.  
Direct functional studies ensure that the constant of interest, $\omega(A,a)^{-1}$, is a decreasing function with respect to $\beta_A$ and to $\beta_a$ whatever the other parameters are (see Fig.~\ref{fig_wp}, left).
Hence, the stronger the sexual preference is, the faster the reproductive isolation is.

Then, I focus on how the constant depends on $m_A$ and $m_a$. It may be natural to consider that $m_A$ and $m_a$ can be rewritten using three positive parameters $\gamma_A$, $\gamma_a$ and $m$ as follows:
$$
m_A:=\gamma_A m \quad \text{and} \quad m_a:=\gamma_a m.
$$
In this way, both migration parameters change simultaneously with $m$. Once again, a direct functional study ensures that $\omega(A,a)^{-1}$ is a non-increasing function with respect to $m$ (see Fig.~\ref{fig_wp}, right). \begin{figure}[!h]
\begin{minipage}{0.49\textwidth}
 \begin{center}
  \includegraphics[width=0.9\textwidth]{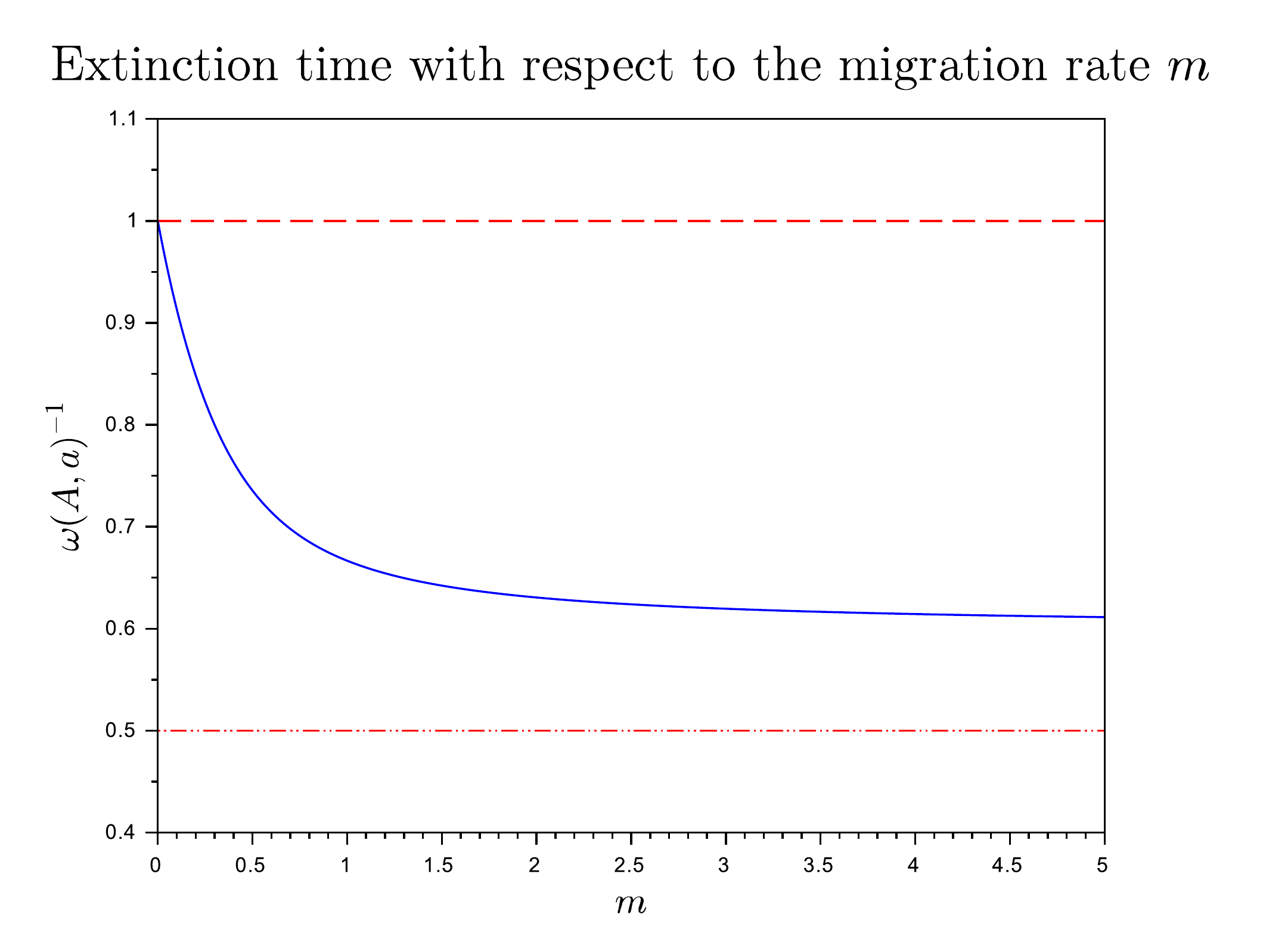}
 \end{center}
 \end{minipage}
 \begin{minipage}{0.49\textwidth}
 \begin{center}
\includegraphics[width=0.9\textwidth]{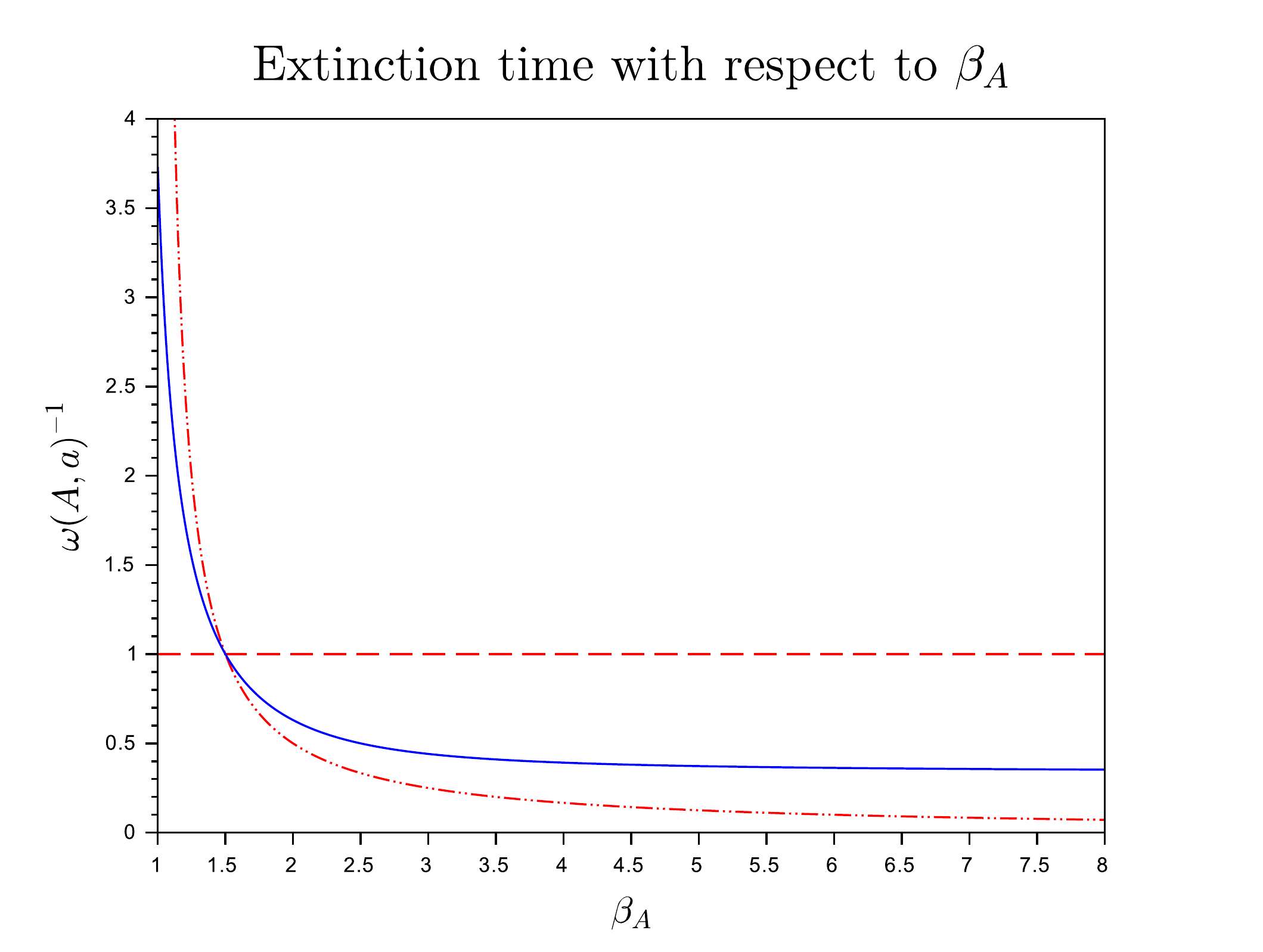}  
 \end{center}
 \end{minipage}
 \caption{\label{fig_wp} 
 Graphs of the constants in front of the times before reproductive isolation, $\omega(A,a)^{-1}$ (blue line), $\omega(A,A)^{-1}$ (red dashed line), $\omega(a,a)^{-1}$ (red dashed-dotted line), with respect to $m$ (left) and to $\beta_A$ (right). The demographic parameters are defined by~\eqref{param_simu}, $\gamma_A=1$, $\gamma_a=\beta_a-1=0.5$ and $\beta_A=2$ on the left and $m=2$ on the right.}
\end{figure}
Hence, increasing both migration rates at the same time accelerates the reproductive isolation, in the same way as when mating preference parameters increase. Moreover, the migration parameters used in the model are frequently-dependent terms such that individuals are more prone to migrate when they do not find suitable mate in their deme. With this in mind, the first conclusion is that a large migration rate seems to strengthen the homogamy.

The result is then improved by studying how constant $\omega(A,a)^{-1}$ changes with respect to $m_A$ and $m_a$ separately. A direct computation shows that $\omega(A,a)^{-1}$ is a decreasing function with respect to $m_A$ if $\beta_A>\beta_a$ and it is an increasing function with respect to $m_A$ if $\beta_A<\beta_a$. In other words, if $A$-individuals have a stronger sexual preference than $a$-individuals ($\beta_A>\beta_a$), the bigger their migration rate is when they are in contact with too much $a$-individuals, the shorter the time before reaching the equilibrium is. Once again, it highlights that the effects of migration and sexual preference are similar. However, assuming again that $A$-individuals have a stronger sexual preference than $a$-individuals ($\beta_A>\beta_a$), the bigger the $a$-individuals migration rate is, the longer the time before reproductive isolation is. This is more surprising. In particular, it highlights that a large migration rate does not only reflect a strong sexual preference but implies more involved behavior. This will be corroborated in what follows.\\

\noindent
\textbf{Basins of attraction} : I then explored how basins of attraction are modified when migration parameters increase. To simplify the study, I assumed here that $m:=m_A=m_a$.\\
\indent Figure~\ref{fig1_asym} presents the trajectories of some solutions of dynamical system~\eqref{systdetasym} in both
phase planes which represent both patches. The trajectories are drawn for the initial condition
$$
z_{A,1}(0)=4, \quad z_{a,1}(0)=10, \quad z_{A,2}(0) = 8.5 \quad \text{and} \quad z_{a,2}(0)=15,
$$
and for three different values of $m$: $ 0, 1$ and $5$.
\begin{figure}[h]
 \begin{minipage}{0.49\textwidth}
\centering
  \includegraphics[width=0.95\textwidth]{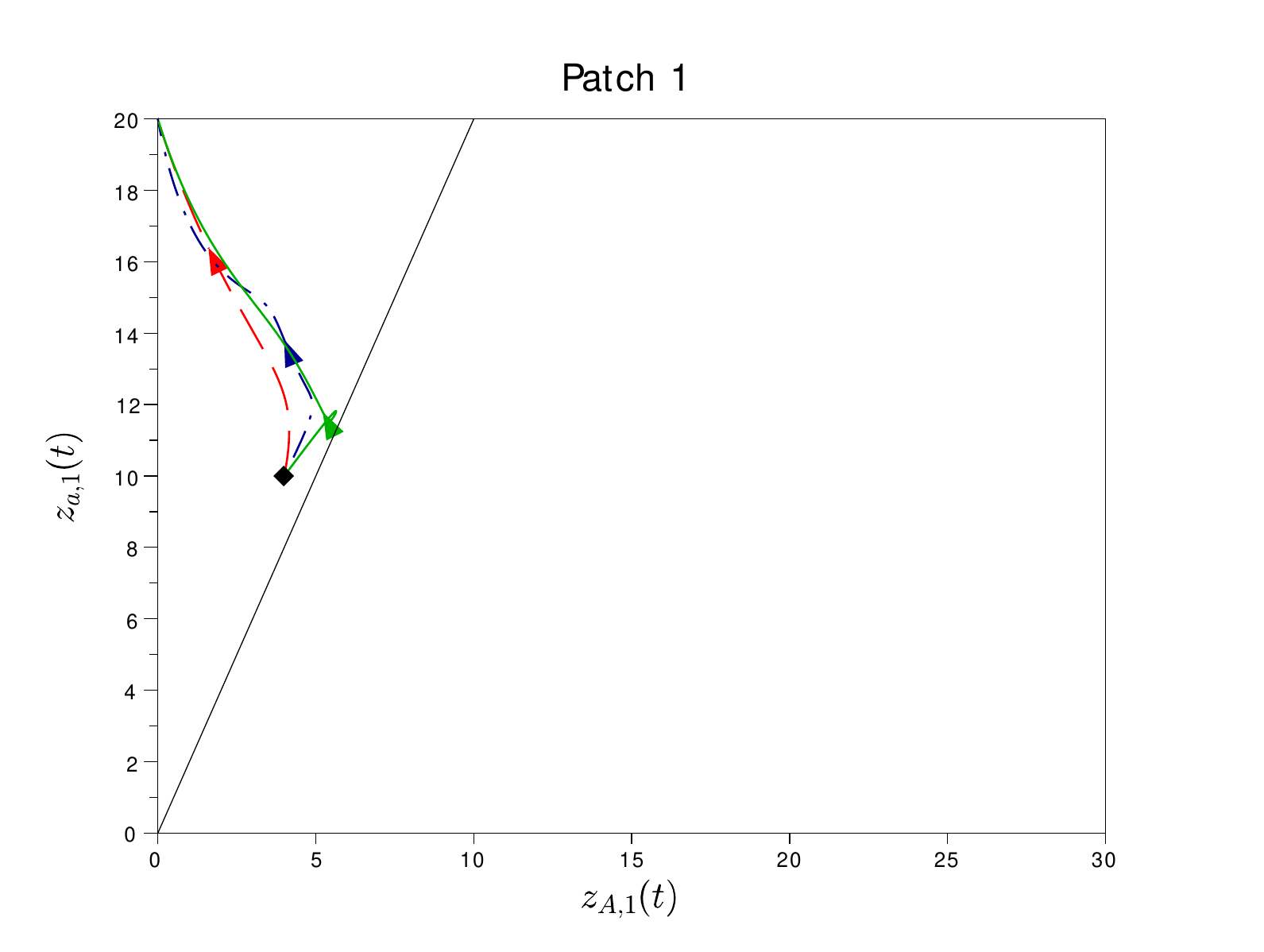}
 \end{minipage}
\begin{minipage}{0.49\textwidth}
\centering
  \includegraphics[width=0.95\textwidth]{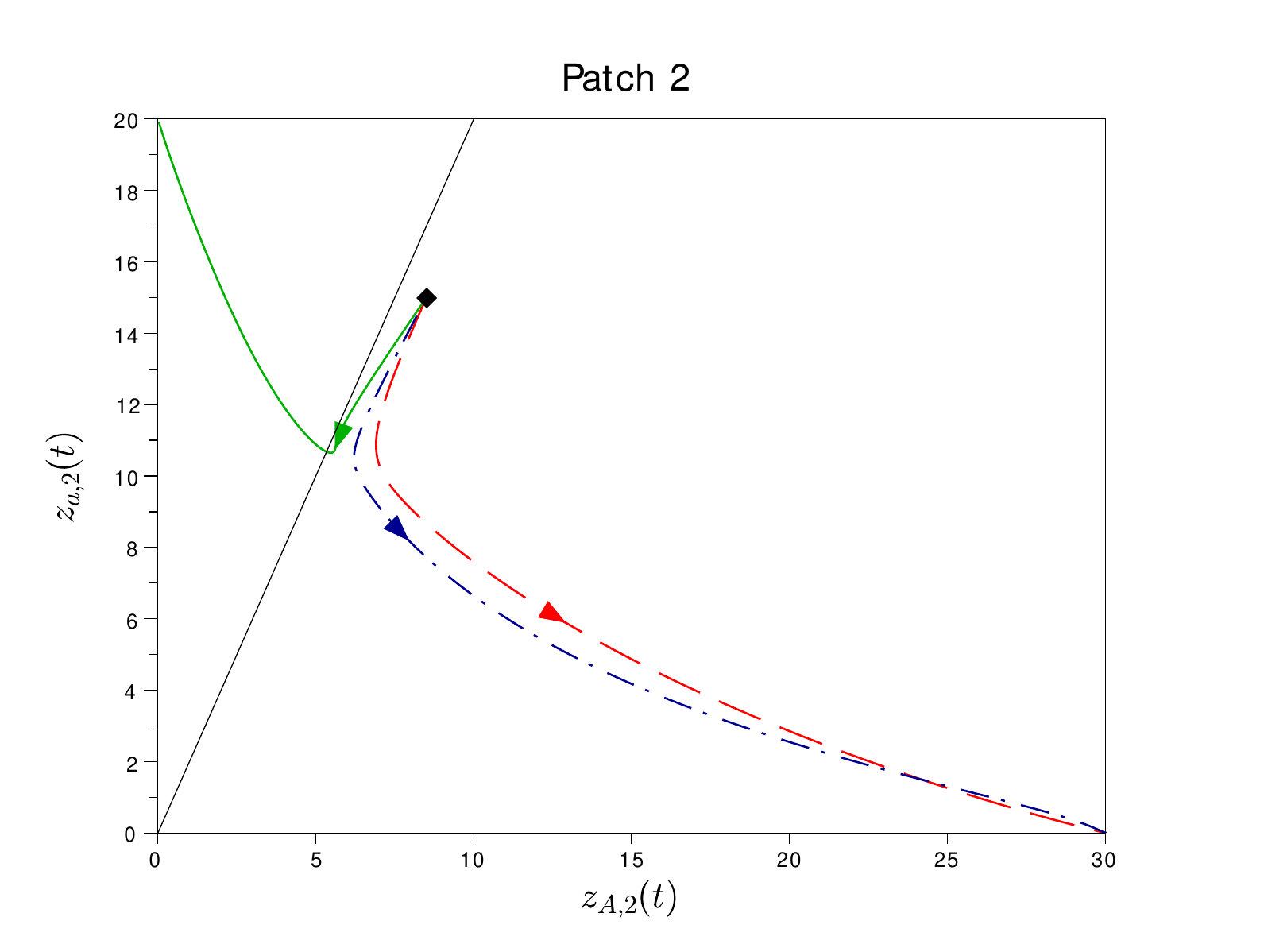}
 \end{minipage}
\caption[Plots of the trajectories in the phase planes]{\label{fig1_asym}\small{Plots of the trajectories in the phase planes  
which represent the patch $1$ (left) and the patch $2$ (right) for $t\in [0,10]$ 
and for three values of $m$: $m=0$ (red), $m=1$ (blue), $m=5$ (green). The initial condition is $(4,10,8.5,15)$, represented by the black dots. The dark line is the solution of $(\beta_A-1)z_{A}-(\beta_a-1)z_{a}=0$}}
\end{figure}
It is important to notice that the equilibrium reached depends not only on the initial condition but also on the value of $m$, unlike the symmetrical case. Indeed, on the example of Figure~\ref{fig1_asym}, when $m$ is small, the trajectory converges to $(0,\zeta_a,\zeta_A,0)$. When $m$ is larger, only $a$-individuals survive, the trajectory converges to $(0,\zeta_a,0,\zeta_a)$. Hence, a large migration rate $m$ can favor allele $a$, which codes for the weakest of both mating preferences ($\beta_a<\beta_A$), to invade both patches. 

 Then, an example of basins of attraction $\mathcal{D}^{\alpha,\alpha'}_{m,m}$ is given in the case of a large migration parameter ($m=5$). Figure~\ref{fig2_asym} presents the projections of the four sets on six different planes. More specifically, each graph (a-f) represents the equilibrium reached with respect to the initial condition in patch $1$ for a couple of initial conditions in patch $2$, which is plotted on graph (g).
In order to compare results for $m=5$ and $m=0$, I plotted the line solution of $(\beta_A-1)z_{A,1}-(\beta_a-1)z_{a,1}=0$ on all planes. Indeed, according to Lemma~\ref{lemma_sanspasym}, without migration any trajectory with initial conditions in patch 1 above (resp. below) this line converges to a patch filled with $a$-individuals (resp. $A$-individuals).
\begin{figure}[ht]
\begin{minipage}{0.76\textwidth}
 
\begin{minipage}{0.32\textwidth}
\centering
  \includegraphics[width=0.99\textwidth]{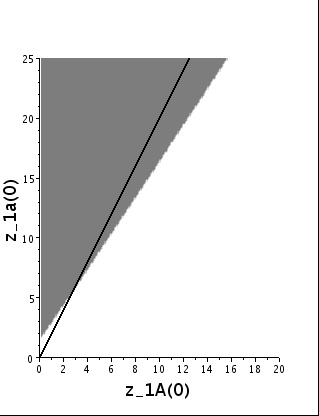}
\small{(a) $z_{A,2}(0)=4$, $z_{a,2}(0)=5 $}\\
 \end{minipage}
\begin{minipage}{0.32\textwidth}
\centering
  \includegraphics[width=0.99\textwidth]{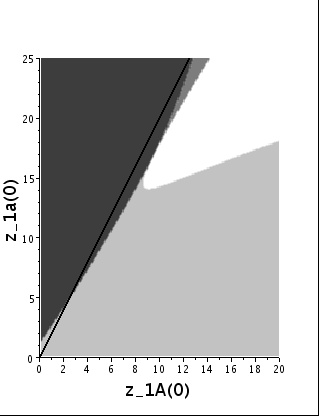}
\small{(b) $z_{A,2}(0)=4$, $z_{a,2}(0)=10$}\\
 \end{minipage}
 \begin{minipage}{0.32\textwidth}
\centering
  \includegraphics[width=0.99\textwidth]{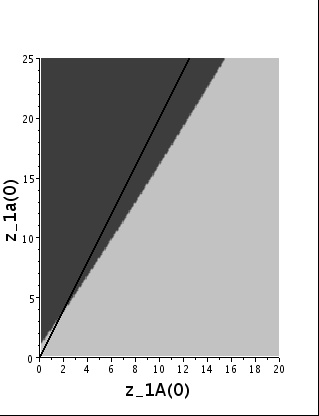}
\small{(c) $z_{A,2}(0)=4$, $z_{a,2}(0)=15$}\\
 \end{minipage}

\begin{minipage}{0.32\textwidth}
\centering
  \includegraphics[width=0.99\textwidth]{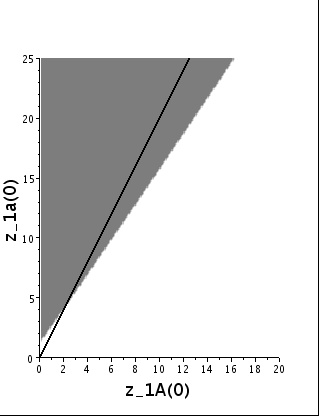}
\small{(d) $z_{A,2}(0)=8$, $z_{a,2}(0)=5 $}\\
 \end{minipage}
\begin{minipage}{0.32\textwidth}
\centering
  \includegraphics[width=0.99\textwidth]{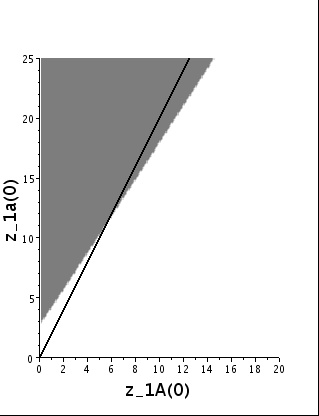}
\small{(e) $z_{A,2}(0)=8$, $z_{a,2}(0)=10$}\\
 \end{minipage}
\begin{minipage}{0.32\textwidth}
\centering
  \includegraphics[width=0.99\textwidth]{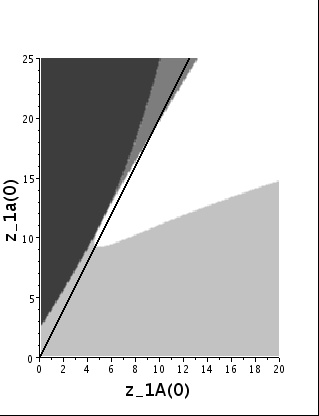}
\small{(f) $z_{A,2}(0)=8$, $z_{a,2}(0)=15$}\\
 \end{minipage}

\end{minipage}
\begin{minipage}{0.23\textwidth}
\centering
  \includegraphics[width=0.99\textwidth]{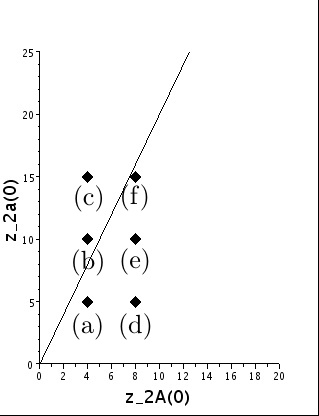}\\
\small{(g) Representation of the initial conditions in the patch 2}\\
 \end{minipage}

\caption[Projections of sets $\mathcal{D}^{\alpha,\alpha'}_m$ with $m=5$]{\label{fig2_asym}\small{(a-f): Projections of sets $\mathcal{D}^{\alpha,\alpha'}_5$ on the planes characterized by the values of $(z_{A,2}(0),z_{a,2}(0))$ given in captions. On each plane, the four sets from white to dark grey corresponds to initial conditions with convergence to $(\zeta_A,0,\zeta_A,0)$, $(\zeta_A,0,0,\zeta_a)$, $(0,\zeta_a,\zeta_A,0)$ and $(0,\zeta_a,0,\zeta_a)$ respectively. The black line is the solution of $(\beta_A-1)z_{A,1}-(\beta_a-1)z_{a,1}=0$. (g): The black diamond points correspond to the initial conditions in patch 2 taken to obtain plots (a) to (f).}}
\end{figure}
Generally, observe that when the number of $a$-individuals is large in patch $1$, these individuals are favored by a large migration rate. 
Thus, the conclusion here is that a large migration parameter $m$ favors the allele coding for the weakest mating preference by mixing the populations of both patches.\\

\noindent
\textbf{Minimal number of individuals for invasion} : 
Initially, each patch is filled with a density of $\zeta_a$ $a$-individuals and $A$-individuals are introduced in patch $1$. 
To corroborate previous observations, I computed the minimal number of $A$-individuals that is needed to be introduced such that they can survive, i.e. such that the dynamical system converges to a stable equilibrium with a positive number of $A$-individuals. I computed this minimal number, denoted by $N^{\min}(\beta_A,m)$, for a range of values of $\beta_A$ ($\beta_A\in (1,2]$) and $m$ ($m\in [0,2]$), other parameters are defined by~\eqref{param_simu}. 
\begin{figure}[!ht]
\begin{center}
 \includegraphics[width=0.99\textwidth]{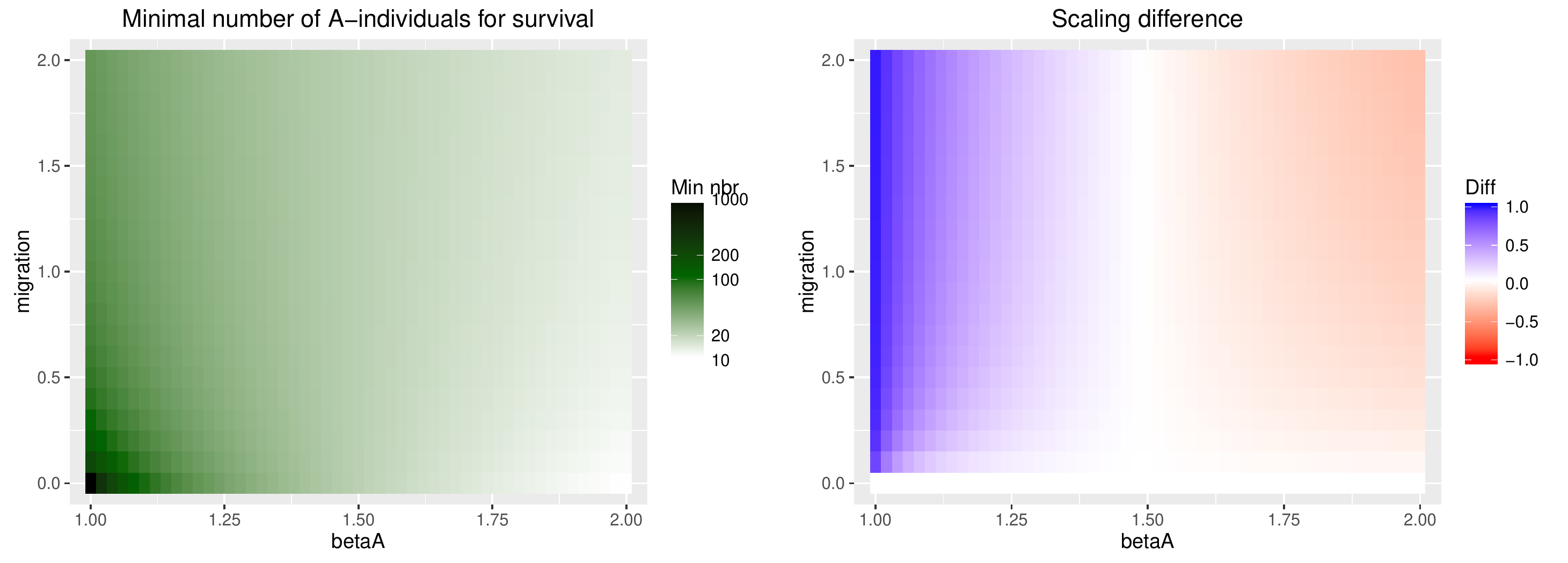}
\end{center}
\caption{\label{fig_minnbr} Left: Minimal number of initial $A$-individuals in patch $1$ that is needed for a long time survival when starting from two patches filled with $\zeta_a$ $a$-individuals;
a logarithmic color scale is used. Right: Scaling differences between the minimal number of $A$-individuals needed without migration (i.e. $(\beta_a-1)\zeta_a/(\beta_A-1)$) and the one computed on the left plot. Parameters are defined by~\eqref{param_simu} where $\beta_a=1.5$.} 
\end{figure}
On the left part of Figure~\ref{fig_minnbr}, the number $N^{\min}(\beta_A,m)$ is drawn using a logarithmic scale. 
Note that the minimal number of $A$-individuals required for survival decreases when $\beta_A$ increases. For example, when $\beta_A$ is large ($\beta_A=2$), observe that the minimal number of $A$-individuals needed for survival, is only half (resp. two-thirds) of $\zeta_a=20$ when $m=0$ (resp. $m=2$). Moreover, if $\beta_A$ and $m$ are sufficiently large ($\beta_A\geq 2.9$ and $m\geq 1.9$ (data not shown)), the $A$-population replaces the resident $a$-population in both patches as soon as the initial number of $A$-individuals is equal to $N^{\min}(\beta_A,m)$. This suggests that individuals with a higher mating preference have a selective advantage.

Secondly, to better understand how $m$ affects $N^{\min}(\beta_A,m)$, the minimal number of $A$-individuals needed for survival, I computed the scaling difference
$$
D(\beta_A,m):=\frac{N^{\min}(\beta_A,0)-N^{\min}(\beta_A,m)}{N^{\min}(\beta_A,0)},
$$
on the right part of Figure~\ref{fig_minnbr}.
Section~\ref{subsec_withoutmig} implies that $N^{\min}(\beta_A,0)=(\beta_a-1)\zeta_a/(\beta_A-1)$. 
For $\beta_A$ and $m$ fixed, a positive value of $D(\beta_A,m)$ indicates that the minimal number of $A$-individuals needed for survival is smaller than in the case without migration, that is to say, the migration favors $A$-individuals, especially if $D(\beta_A,m)$ is large. The opposite conclusion holds for negative value of $D(\beta_A,m)$.
Here, when $\beta_A$ is smaller than $\beta_a=1.5$, $D(\beta_A,m)$ is positive and increases with migration $m$ whereas, when it is smaller than $\beta_a$, it is decreasing with $m$. Hence, migration seems here again to favor the allele with the weakest mating preference.\\

\noindent
\textbf{Discussion}: The first conclusion is that a population with a large mating preference has selective advantages: (1) the larger the mating preference strength is, the smaller the time before reaching an equilibrium where this allele is maintained is, and (2) a population with a strong mating preference can invade a resident population presenting a weak preference even if its initial number of individuals is small. Same kind of conclusion is drawn by \cite{smadja_catalan_2004}. In the latter, the authors predict that the asymmetrical mating preference observed between two species of mouse could lead to
 the replacement of the subspecies with the weakest mating preference
 (M. m. domesticus) by the other subspecies
 (M. m. musculus), if no other mechanism was involved. 
This conclusion is a substantial added value compared to \cite{coron2016stochastic} where only the symmetric case ($\beta_A=\beta_a$, $m_A=m_a$) is considered. Accounting for asymmetrical preference gave the possibility to better understand advantages of a strong mating preference.

Migration has a more involved impact on the system dynamics than mating preference, although the frequency-dependent term I used for migration seemed only to mimic mating preferences. More precisely, there exists a trade-off between two phenomena \cite{coron2016stochastic}:
 (1) large migration rates can help individuals to escape disadvantageous patches \cite{clobert2001dispersal} but (2) large migration rates entail also risks of moving to unfamiliar patches (i.e. filled with not-preferred individuals) and thus may increase the time before reproductive isolation. This is understandable since the migration terms only focus on the departure patch. 
 More surprisingly,
large migration rates seem to favor alleles with reduced mating preferences. 
 This tendency was not noticed by \cite{coron2016stochastic} and could be linked to the effects of migration on habitat specialization \cite{brown1992evolution,cuevas2003evolution,elena2009evolution}. In these articles, the authors highlight that migration may prevent the local specialization of subpopulations and favor generalist species. Hence, in both cases, large migration rates tend to avoid specialized behaviors in terms of ecological niches adaptation or mating partners adaptation.

\section{Proofs}

\label{subsec_perturbation}
This last part is devoted to the proof of Theorem~\ref{theo_sysasym}.
The main idea is to start from the results without migration, then use a perturbation method to make $m_A$ and $m_a$ grow up and deduce results for some positive migration parameters. 

However, this perturbation technique will only apply on a bounded set of $\R^{\mathcal{E}}$ excluding $0$. Thus, let us first prove Lemma~\ref{lemma_invsetasym}, which allows us to restrict the study of the dynamical system~\eqref{systdetasym} to the compact set $\mathcal{S}$. 
To help with proofreading, we recall here the definitions of the weighted sums :
 \begin{eqnarray*}
 &\Sigma_i:= (\beta_A-1)z_{A,i}+(\beta_a-1)z_{a,i}, \text{ for } i=1,2,\\
 &\Sigma:= \Sigma_1+\Sigma_2 = (\beta_A-1)(z_{A,1}+z_{A,2})+(\beta_a-1)(z_{a,1}+z_{a,2}).
 \end{eqnarray*}

\begin{proof}[Proof of Lemma~\ref{lemma_invsetasym}]
 The proof is based on the equations satisfied by $\Sigma_1$, $\Sigma_2$ and $\Sigma$. From~\eqref{systdetasym}, we find
\begin{multline}
\label{eq_sigma1}
 \frac{d}{dt}\Sigma_1=\Sigma_1 \left[ b\frac{\Sigma_1}{z_{A,1}+z_{a,1}} -2b(\beta_A-1)(\beta_a-1)\frac{z_{a,1}z_{A,1}}{(z_{A,1}+z_{a,1})\Sigma_1} +b-d-c(z_{A,1}+z_{a,1}) \right]\\
-\big(m_A(\beta_A-1)+ m_a(\beta_a-1)\big)\left[\frac{z_{A,1}z_{a,1}}{z_{A,1}+z_{a,1}}-\frac{z_{A,2}z_{a,2}}{z_{A,2}+z_{a,2}}\right].
\end{multline}
Since $\Sigma_1^2-2(\beta_A-1)(\beta_a-1)z_{a,1}z_{A,1}\geq 0$ and $\Sigma_1\geq (\beta_{\min}-1)(z_{a,1}+z_{A,1})$, we obtain
\begin{equation}
\label{eq_sigma1geq}
 \frac{d}{dt}\Sigma_1\geq \Sigma_1 \left[ b-d-\frac{c}{(\beta_{\min}-1)}\Sigma_1 -\big(m_A(\beta_A-1)+ m_a(\beta_a-1)\big)\frac{z_{A,1}z_{a,1}}{(z_{A,1}+z_{a,1})\Sigma_1} \right].
\end{equation}
We then find an upper bound of $\frac{z_{A,1}z_{a,1}}{(z_{A,1}+z_{a,1})\Sigma_1}$:
\begin{align*}
 \Sigma_1 (z_{A,1}+z_{a,1})&=(\beta_A-1)z_{A,1}^2+(\beta_a-1)z_{a,1}^2+(\beta_A+\beta_a-2)z_{A,1}z_{a,1}\\
 & \geq (\beta_{\min} -1)[z_{A,1}^2+z_{a,1}^2+2 z_{A,1}z_{a,1}]\\
 & \geq 4(\beta_{\min} -1)z_{A,1}z_{a,1}.
\end{align*}
In addition with~\eqref{ass_pasym} and~\eqref{eq_sigma1geq}, we deduce
\begin{equation*}
 \frac{d}{dt}\Sigma_1\geq \Sigma_1 \left[ \frac{b-d}{2}-\frac{c}{(\beta_{\min}-1)}\Sigma_1 \right].
\end{equation*}
Hence, as soon as $\Sigma_1<{(\beta_{\min} -1)(b-d)}/{2c}$, its derivative is strictly positive. In other words, if $\Sigma_1(0)\leq {(\beta_{\min} -1)(b-d)}/{4c}$, there exists $t_1>0$ such that for all $t\geq t_1$, $\Sigma_1(t)$ is higher than this threshold. Moreover, if $\Sigma_1(t_2)$ is higher than this threshold, for all $t\geq t_2$, $\Sigma_1(t)$ remains higher than it. The same conclusion holds for $\Sigma_2$.\\
Let us now deal with $\Sigma$. From equations~\eqref{eq_sigma1} satisfied by $\Sigma_1$ and $\Sigma_2$, we find
\begin{align*}
 \frac{d}{dt}\Sigma&=\sum_{i=1,2}\Sigma_i \left[ b\frac{\beta_Az_{A,i}+\beta_a z_{a,i}}{z_{A,i}+z_{a,i}}-d-c(z_{A,i}+z_{a,i})\right]-2b(\beta_A-1)(\beta_a-1)\frac{z_{A,i}z_{a,i}}{z_{A,i}+z_{a,i}}\\
&\leq \sum_{i=1,2}\Sigma_i \left[ b\beta_{\max}-d-\frac{c}{\beta_{\max} -1}\Sigma_i\right]\\
&\leq \Sigma \left[ b\beta_{\max}-d-\frac{c}{2(\beta_{\max} -1)}\Sigma\right].
\end{align*}
Using a reasoning similar to the previous one, we conclude that there exists a time after which $\Sigma(t)$ remains lower than $4(\beta_{\max}-1)(b\beta_{\max}-d)/c$. Finally any trajectory hits $\mathcal{S}$ after a finite time and $\mathcal{S}$ is a positively invariant set. That ends the proof of Lemma~\ref{lemma_invsetasym}.

\end{proof}

Lemma~\ref{lemma_invsetasym} implies that the study of the dynamical system~\eqref{systdetasym} can be restricted to the study of trajectories belonging to $\mathcal{S}$. Note that when $m_A=m_a=0$, Subsection~\ref{subsec_withoutmig} ensures that the dynamical system~\eqref{systdetasym} has exactly 9 equilibria which belong to $\mathcal{S}$:
\begin{eqnarray}
& \label{eqbre_attractive} (\zeta_A,0,0,\zeta_a),\; (\zeta_A,0,\zeta_A,0),\;  (0,\zeta_a,\zeta_A,0),\; (0,\zeta_a,0,\zeta_a),\\
& \label{eqbre_saddle3} (\chi_A,\chi_a,\zeta_A,0),\; (\chi_A,\chi_a,0,\zeta_a),\; (0,\zeta_a,\chi_A,\chi_a),\; (\zeta_A,0,\chi_A,\chi_a).\\
& \label{eqbre_saddle2} (\chi_A,\chi_a,\chi_A,\chi_a).
\end{eqnarray}
Equilibria~\eqref{eqbre_attractive} are stable fixed point whereas equilibria~\eqref{eqbre_saddle3} (resp.~\eqref{eqbre_saddle2}) are unstable with a local stable manifold of dimension $3$ (resp. $2$), i.e. there exists an empty interior set of dimension $3$ (resp. $2$) such that any trajectory starting from this set converges to equilibria~\eqref{eqbre_saddle3} or \eqref{eqbre_saddle2}.\\

In order to simplify the notation of the proofs, let us write the migration rates $m_A$ and $m_a$ using three parameters $\gamma_A\in [0,1]$, $\gamma_a\in [0,1]$ and $m\geq 0$ as
$$
m_A := m \gamma_A \quad \text{and} \quad m_a:=m \gamma_a.
$$
We consider that $\gamma_A$ and $\gamma_a$ are fixed parameters and we will make $m$ grow up in the following proof.
We can rewrite the dynamical system~\eqref{systdetasym} considering $m$ as a parameter
\begin{equation}
\label{eq_systF}
 \frac{d}{dt}\textbf{z}(t)=F(\textbf{z}(t),m).
\end{equation}
The solution of~\eqref{eq_systF} with initial condition ${\bf{z}}^0$ is written $t \mapsto \varphi_{m,\bf{z}^0}(t)$.
Our goal is to understand the dynamics of the flow $\varphi_{m,\bf{z}^0}$ associated to the vector field $F({\bf{z}},m)$ using $\varphi_{0,\bf{z}^0}$ (without migration) which is entirely described in Subsection~\ref{subsec_sansp}. Theorem~\ref{theo_sysasym} can be rewritten as follows using the notion of flow.

\begin{theorem}[Theorem~1']
\label{lemma_convsetasym}
There exists $m_0>0$ such that for all $m\leq m_0$, we can find four open subsets $(\mathcal{D}^{\alpha,\alpha'}_m)_{\alpha,\alpha'\in \CA}$ of $\mathcal{S}$ with the following properties:
\begin{itemize}
 \item The closure of $\cup_{\alpha,\alpha'\in \CA} \mathcal{D}^{\alpha,\alpha'}_m$ is equal to $\mathcal{S}$.
 \item For all ${\bf{z}^0}\in \mathcal{D}^{A,a}_m$, the flow $\varphi_{m,\bf{z}^0}(t)$ converges to $(\zeta_A,0,0,\zeta_a)$ when $t$ tends to $+\infty$. Similar results hold for the three other equilibria~\eqref{eqbre_attractive}.\\
\end{itemize}
\end{theorem}

\begin{proof}
The first step is to construct a neighborhood around each equilibrium \eqref{eqbre_attractive}-\eqref{eqbre_saddle2} which includes a unique equilibrium of the dynamical system \eqref{systdetasym} with $m>0$.\\
Let us first focus our study on the equilibrium $(\zeta_A,0,0,\zeta_a)$. Subsection~\eqref{subsec_sansp} implies that, when $m=0$, the equilibrium $(\zeta_A,0,0,\zeta_a)$ is an attractive stable equilibrium.\\
The first derivative $D_{{\tb{z}}}F$ evaluated at $({\tb{z}},m)=((\zeta_A,0,0,\zeta_a),0)$ is 
\begin{equation}
\label{eq_matjac}
\begin{pmatrix}
 -(b\beta_A-d) & -b(\beta_A-1)-(b\beta_A-d) & 0 &0 \\
0 & -b(\beta_A-1) & 0 & 0\\
0 & 0 & -b(\beta_a-1) & 0\\
0 &0 & -b(\beta_a-1)-(b\beta_a-d) & -(b\beta_a-d)
\end{pmatrix}.
\end{equation}
Since matrix~\eqref{eq_matjac} is invertible and $F$ is smooth on $\mathcal{S}\times \R^+$, the Implicit Function Theorem insures that there exists $m_1$ and a neighborhood $\mathcal{V}_1$ of $(\zeta_A,0,0,\zeta_a)$ in $\mathcal{S}$ such that there is a unique point ${\tb{y}}_1(m)\in \mathcal{V}_1$ satisfying $F(\tb{y}_1(m),m)=0$ for all $m<m_1$. And $m\mapsto {\tb{y}}_1(m)$ is regular and converges to $(\zeta_A,0,0,\zeta_a)$ when $m$ converges to $0$. 
A simple computation ensures that $F(\tb{y}_1(0),m)=F((\zeta_A,0,0,\zeta_a),m)=0$, for any $m>0$. Since ${\tb{y}}_1(m)$ is unique, we deduce that ${\tb{y}}_1(m)={\tb{y}}_1(0)$. \\
Moreover, from Theorem 6.1 and Section 6.3 of \cite{ruelle1989book} (see also Appendice B of \cite{collet2011rigorous}, or \cite{hoppensteadt1966singular}), we conclude that if $m_1$ and $\mathcal{V}_1$ are small enough, any solution $\varphi_{m,\tb{z}^0}$ with ${\tb{z}}^0 \in \mathcal{V}_1$ and $m<m_1$ converges uniformly to $\varphi_{0,\tb{z}^0}$ when $m$ converges to $0$. In other words, ${\tb{y}}_1(0)$ attracts all the orbits $\varphi_{m,\tb{z}^0}$ starting from $\mathcal{V}_1$.\\
Similarly, we find $(m_i)_{i=2,3,4}$ and $(\mathcal{V}_i)_{i=2,3,4}$ neighborhoods around the three other equilibria of~\eqref{eqbre_attractive}, denoted by $(\tb{y}_i(0))_{i={2,3,4}}$, such that, for $i\in \{2,3,4\}$, for all $m<m_i$, $\tb{y}_i(0)$ attracts all solutions $\varphi_{m,\tb{z}^0}$ with ${\tb{z}}^0 \in \mathcal{V}_i$ and $m<m_i$. \\
Theorem 6.1 and Section 6.3 of \cite{ruelle1989book} ensure also the stability of the local stable and unstable manifolds of a hyperbolic non-attractive fixed points. Thus, we find $m_5,..,m_9$ and $\mathcal{V}_5,..,\mathcal{V}_9$, neighborhoods around equilibria~\eqref{eqbre_saddle3} and~\eqref{eqbre_saddle2} 
that satisfy the following properties. For all $i\in \{5,..,9\}$, for all $m<m_i$, there exists a unique fixed point $\mathbf{y}_i(m)\in\mathcal{V}_i$ invariant by $F(.,m)$ which repulses all orbits solution associated to $F(.,m)$, except the orbits that start from a surface of dimension $3$ or $2$, depending on whether we are focused on equilibria~\eqref{eqbre_saddle3} or~\eqref{eqbre_saddle2} respectively. These surfaces are the stable manifolds of $(\mathbf{y}_i(m))_{i=5,..,9}$ in $(\mathcal{V}_i)_{i=5,..,9}$ respectively. \\
Without loss of generality, we assume that these nine neighborhoods are disjoint sets.\\

The second step is to deal with trajectories outside these nine neighborhoods.
Let $\eps>0$ and for $i=1,..,9$, we define
$$
\mathcal{V}_i^\eps=B(\tb{y}_i(0), R_i)=\{{\bf{z}}\in \mathcal{S}, \|{\bf{z}}-{\bf{y}}_i(0)\|\leq R_i\}, \quad \text{where} \quad R_i=\max\{r>0, B(\tb{y}_i(0),r+\eps)\subset \mathcal{V}_i\},
$$
which is a neighborhood of $\tb{y}_i(0)$ slightly smaller than $\mathcal{V}_i$.\\
Recall that the five neighborhoods $(\mathcal{V}_i^\eps)_{i=5,..,9}$ attracts some solutions $\varphi_{0,\mathbf{z}^0}$. Thus, we set
\begin{equation}
\label{def_Wset}
\mathcal{W}=   \left( \underset{i=5,..,9}{\bigcup} \; \underset{\mathbf{z}^0\in \mathcal{V}_i^\eps}{\bigcup} \; (\varphi_{0,\mathbf{z}^0})^{-1}([0,+\infty)) \, \right) \bigcap \mathcal{S},
\end{equation}
which is a neighborhood of the union of all stable manifolds of unstable equilibria~\eqref{eqbre_saddle2} and~\eqref{eqbre_saddle3} assuming $m=0$. We denote the complement of $\mathcal{W}$ in $\mathcal{S}$ by $\mathcal{W}^c$.

Let us first deal with the trajectories starting from $\mathcal{W}^c$.
According to Appendix~\ref{subsec_sansp}, all trajectories $\varphi_{0,\tb{z}^0}$ starting from $\mathcal{W}^c$ converge to a stable equilibrium, i.e. they reach any neighborhood of set $\{\tb{y}_i(0),i=1,.,4\}$ in finite time.
Since $\mathcal{W}^c$ is compact, there exists a finite time $t_1>0$ such that $\varphi_{0,\tb{z}^0}(t_1) \in \cup_{i=1}^4 \mathcal{V}_i^\eps$, for all $\tb{z}^0\in \mathcal{W}^c$.
Moreover, from Theorem 1.4.7 by \cite{berger1992geometrie}, the flow $\varphi$ is uniformly continuous with respect to $m$, to ${\tb{z}}^0$ and to $t$. We can thus find $m_{10}<\min_{i=1,..,9}m_i$ such that for every $m\leq m_{10}$, $\tb{z}^0\in \mathcal{W}^c$
\begin{equation*}
 \label{eq_proche}
\left\vert \varphi_{0,\tb{z}^0}(t_1)-\varphi_{m,\tb{z}^0}(t_1) \right\vert \leq \eps.
\end{equation*}
Then, by definition of $(\mathcal{V}_i)_{i=1,..,4}$ and $(\mathcal{V}^\eps_i)_{i=1,..,4}$, we deduce that for all $m\leq m_{10}$, all $\tb{z}^0\in\mathcal{W}^c$ and all $t\geq t_1$,
$$
\varphi_{m,\tb{z}^0}(t) \in \bigcup_{i=1}^4 \mathcal{V}_i.
$$

Then we deal with the trajectories starting from $\mathcal{W}$.
According to the definition of $\mathcal{W}$~\eqref{def_Wset}, all trajectories $\varphi_{0,\tb{z}^0}$ starting from $\mathcal{W}$ reach one of the five neighborhoods $(\mathcal{V}_i^\eps)_{i=5,..,9}$ in finite time. Thus, by reasoning as above, we can find $m_{11} \leq m_{10}$ and $t_2 >0$ such that for all $m\leq m_{11}$ and all $\tb{z}^0\in \mathcal{W}$, there exists $t\leq t_2$, with
$$
\varphi_{m,\tb{z}^0}(t) \in \bigcup_{i=5}^9 \mathcal{V}_i.
$$
Let us fix $ m\leq m_{11}$, $\tb{z}^0\in \mathcal{W}$ and assume that $\varphi_{m,\tb{z}^0}(t_3) \in \mathcal{V}_i$. We have then three possibilities:
\begin{itemize}
\item[(i)] If $\varphi_{m,\tb{z}^0}(t) \in \mathcal{V}_i$ for all $t\geq t_3$, then $\tb{z}^0$ belongs to the stable manifold of $\tb{y}_i(m)$ in $\mathcal{S}$. Since we have a global diffeomorphism on $\mathcal{S}$, we can find the stable manifold of $\tb{y}_i(m)$ by iterating the Implicit Function Theorem and deduce that this stable manifold is an empty interior set of dimension $3$ or $2$, depending on which equilibrium is considered.
\item[(ii)] Otherwise, there exists $t_4\geq t_3$ such that $\varphi_{m,\tb{z}^0}(t_4) \not\in \mathcal{V}_i$. If $\varphi_{m,\tb{z}^0}(t_4) \in \mathcal{W}^c$, the flow will converge to one of the four equilibria~\eqref{eqbre_attractive} according to previous reasoning.
\item[(iii)] The last possibility is $\varphi_{m,\tb{z}^0}(t_4) \in \mathcal{W}\setminus \cup_{i=5}^9 \mathcal{V}_i$. Thus, the flow $(\varphi_{m,\tb{z}^0}(t))_{t\geq t_4}$ will reach again one of the neighborhoods $(\mathcal{V}_j)_{j=5,..,9}$. It would have a problem if the trajectory went from a neighborhood to an other without living $\mathcal{W}$ as $t \mapsto +\infty$. Thus, let us show that this is not possible. Indeed, the flow goes out of $\mathcal{V}_i$ by following the unstable manifold of $\tb{y}_i(m)$ which is close to the unstable manifold of $\tb{y}_i(0)$ (according to the continuity of the unstable manifolds with respect to $m$, cf Theorem 6.1 by \cite{ruelle1989book}). Since $\varphi_{m,\tb{z}^0}$ leaves $\mathcal{V}_i$ by staying in $\mathcal{W}$, the intersection of the unstable manifold of $\tb{y}_i(0)$ and $\mathcal{W}$ is not empty. From the definition of $\mathcal{W}$~\eqref{def_Wset} and Appendix~\ref{subsec_sansp}, it is possible if and only if $\tb{y}_i(0)=\tb{y}_9(0)=(\chi_A,\chi_a,\chi_A,\chi_a)$ and if $\varphi_{m,\tb{z}^0}$ leaves $\mathcal{V}_9$ through the neighborhood of the stable manifold of one of the equilibria~\eqref{eqbre_saddle3}. Thus, the flow $(\varphi_{m,\tb{z}^0}(t))_{t\geq t_4}$ will reach one of the neighborhood $(\mathcal{V}_j)_{j \in \{5,6,7,8\}}$. Then, only the two previous possibilities (i) or (ii) are possible. 
\end{itemize}
Finally, we have shown that any solution $\varphi_{m,\tb{z}^0}$ of~\eqref{eq_systF} starting from $\mathcal{S}$ and with $m\leq m_{11}$ converges to one of the equilibria~\eqref{eqbre_attractive}, except if it starts from a set with empty interior which is the union of the global stable manifolds of the equilibria $(\tb{y}_i(m))_{i=5,..,9}$.\\

Finally, $m_0:=m_{11}$, 
\begin{equation*}
\mathcal{D}^{A,a}_m=\underset{{\tb{z}}^0\in \mathcal{V}_1}{\cup} \; {\varphi_{m,{\tb{z}}^0}}^{-1}([0,+\infty)),
\end{equation*}
and $\mathcal{D}_m^{A,A}$, $\mathcal{D}_m^{a,A}$, $\mathcal{D}_m^{a,a}$ are defined in a similar way using sets $\mathcal{V}_2$, $\mathcal{V}_3$ and $\mathcal{V}_4$ respectively.
We have shown that for all $m\leq m_0$, the four non empty interior sets $(\mathcal{D}^{\alpha,\alpha'}_m)_{\alpha,\alpha'=A,a}$ satisfy Theorem~\ref{lemma_convsetasym}.
\end{proof}

\appendix

\section{Dynamical system without migration}
\label{subsec_sansp}

In this appendix, we will prove the results of Section~\ref{subsec_withoutmig}, which is related to the case without migration. To this aim, we use the two following weighted quantities
\begin{equation*}
 \label{def_diffpopasym}
\Omega(t):=(\beta_A-1)z_A(t)-(\beta_a-1)z_a(t),
\end{equation*}
\begin{equation*}
 \label{def_sumpopasym}
\Sigma(t):=(\beta_A-1)z_A(t)+(\beta_a-1)z_a(t).
\end{equation*}
From~\eqref{eq_system1patch}, we find that
\begin{eqnarray}
 \label{eq_diffpopasym} &\frac{d}{dt}\Omega(t)=\Omega \left[ b\dfrac{\beta_Az_A+\beta_a z_a}{z_A+z_a}-d-c(z_A+z_a) \right],\\
\label{eq_sumpopasym} &\frac{d}{dt}\Sigma(t)=\Sigma \left[ b\dfrac{\beta_Az_A+\beta_a z_a}{z_A+z_a}-d-c(z_A+z_a) \right] -2b(\beta_A-1)(\beta_a-1)\dfrac{z_az_A}{z_A+z_a}.
\end{eqnarray}

\begin{proof}[Proof of Lemma~\ref{lemma_sanspasym}]
We start by studying the stability of equilibrium $(0,0)$. 
 Assume that $\Sigma(0)>0$. From~\eqref{eq_sumpopasym}, we derive
\begin{equation*}
 \frac{d}{dt}\Sigma \geq \Sigma\left[ b-d +b\left( \frac{\Sigma}{z_A+z_a}-2(\beta_A-1)(\beta_a-1)\frac{z_az_A}{(z_A+z_a)\Sigma} \right)-c(z_A+z_a)\right].
\end{equation*}
Since $\Sigma^2-2(\beta_A-1)(\beta_a-1)z_az_A \geq 0$ and $-(\beta_{\min}-1)(z_A+z_a) \geq -\Sigma$, we deduce that
\begin{equation*}
 \frac{d}{dt}\Sigma \geq \Sigma\left[ b-d -c\frac{\Sigma}{(\beta_{\min}-1)}\right].
\end{equation*}
Hence, as long as $\Sigma\in]0,(b-d)(\beta_{\min}-1)/c[$, $\Sigma(t)$ is increasing. Thus $(0,0)$ is an unstable equilibrium.

The stability of the three other equilibria, $(\zeta_A,0)$, $(0,\zeta_a)$ and $(\chi_A,\chi_a)$, can be deduce by a direct computation of Jacobian matrices at these points, which we do not detail.

Finally, let us study the long time behavior of any solution. Equation~\eqref{eq_diffpopasym} implies that the sign of $\Omega(t)$ is equal at all time and, that $\mathcal{D}^A_0$
is a positively invariant set under dynamical system~\eqref{eq_system1patch}. Moreover, there exists only a stable equilibrium that belongs to the set $\mathcal{D}^A_0$, which is $(\zeta_A,0)$.

We consider the function $W: \mathcal{D}^A_0 \to \R$:
\begin{equation}
 \label{eq_lyapasym}
W(z_A,z_a):=\ln\left( \dfrac{\Sigma}{\Omega} \right)=\ln\left( \dfrac{(\beta_A-1)z_A+(\beta_a-1)z_a}{(\beta_A-1)z_A-(\beta_a-1)z_a}\right)  \geq  0  .
\end{equation}
From~\eqref{eq_diffpopasym} and~\eqref{eq_sumpopasym}, we deduce that
\begin{equation*}
 \label{eq_derivationV}
\dfrac{dW(z_A(t),z_a(t))}{dt}= -2b(\beta_A-1)(\beta_a-1)\dfrac{z_az_A}{(z_A+z_a)\Sigma}  \leq 0 .
\end{equation*}
Moreover for any $(z_A,z_a)\in \mathcal{D}^A_0$, $W(z_A,z_a)=0$ if and only if $z_a=0$. $W(z_A,z_a)$ converges to $+\infty$ when $(\beta_A-1)z_A-(\beta_a-1)z_a$ converges to $0$ and $\frac{dW}{dt}$ is non-positive on $\mathcal{D}^A_0$ and is equal to zero if and only if $z_a=0$.
It ensures that $W$ is a Lyapunov function for~\eqref{eq_system1patch} on the set $\mathcal{D}^A_0$ which cancels only on $\mathcal{D}^A_0 \cap \{z_a=0\}$.
Furthermore, a simple computation gives that the largest invariant set in $\mathcal{D}^A_0 \cap \{z_a=0\}$ is $\{(\zeta_A,0)\}$.
 Theorem 1 of~\cite{lasalle1960some} is thus sufficient to conclude that any solution of~\eqref{eq_system1patch} with initial condition in $\mathcal{D}^A_0$ converges to $(\zeta_A,0)$ when $t$ tends to $+\infty$.
Similarly, we prove that any solution with initial condition in
$
 \mathcal{D}^a_0 
$
converges to $(0,\zeta_a)$.

Finally, assume that $\Omega(0)=0$. Then, $\Omega(t)=0$ for all $t\geq 0$ according to~\eqref{eq_diffpopasym} and, in addition with~\eqref{eq_system1patch}, we derive for all $\alpha\in \CA$,
\begin{equation*}
 \dfrac{d}{dt}z_\alpha=z_\alpha \left[ b\frac{\beta_A\beta_a-1}{\beta_A+\beta_a-2}-d-c\frac{\beta_A+\beta_a-2}{\beta_{\bar\alpha}-1}z_{\alpha} \right].
\end{equation*}
We deduce the last point of Lemma~\ref{lemma_sanspasym} easily.
\end{proof}

\section{Extinction time}
\label{subsec_probaasym}
This subsection is devoted to the proof of Theorem~\ref{theo_asym} following ideas similar to the ones of the proof of Theorem 3 and Proposition 4.1 in \cite{coron2016stochastic}. Hence, we do not give all details, but explain only parts that are different. 

Assume that $m_A\leq m_0$, $m_a\leq m_0$ and that ${\bf{Z}}^K(0)$ converges in probability to a deterministic vector ${\bf{z}^0}$ belonging to $\mathcal{D}^{A,a}_{m_A,m_a}$, Lemma~\ref{lemma_convsetasym} and Theorem~\ref{theo_sysasym} ensure that $({\bf{Z}}^K(t),t\geq 0)$ reaches a neighborhood of the equilibrium $(\zeta_A,0,0,\zeta_a)$ after a finite time independent from $K$. Indeed, the process dynamics is close to the one of the limiting deterministic system~\eqref{systdetasym}. 

To prove Theorem~\ref{theo_asym}, it remains to estimate the time before all $a$-individuals in patch $1$ and all $A$-individuals in patch $2$ disappear. We denote it by
\begin{equation}
 T^K_0=\inf\{t\geq 0, Z^K_{a,1}(t)+Z^K_{A,2}(t)=0\},
\end{equation}
and we assume that the process is initially close to equilibrium $(\zeta_A,0,0,\zeta_a)$. The estimation is deduced from the following Lemma.
\begin{lemma}
 \label{lem_stoasym}
There exist two positive constants $\varepsilon_0$ and $C_0$ such that for any $\varepsilon\leq \varepsilon_0$,
 if there exists $\eta\in ]0,1/2[$ that satisfies $\max(|z_{A,1}^0-\zeta_A|,|z_{a,2}^0-\zeta_a|) \leq \eps$ and
 $\eta\eps/2 \leq z_{a,1}^0,z_{A,2}^0 \leq \eps/2$, then 
 $$
\begin{aligned}
 &\text{for all } C>(\omega(A,a))^{-1}+C_0\eps, &\P(T_0^K\leq C \log(K)) \underset{K\to +\infty}{\to} 1,\\
 &\text{for all } 0\leq C <(\omega(A,a))^{-1}-C_0\eps, & \P(T_0^K\leq C \log(K)) \underset{K\to +\infty}{\to} 0.
\end{aligned}
$$
\end{lemma}

\begin{proof}
 Following the first step of Proposition 4.1's proof given by \cite{coron2016stochastic}, we prove that as long as the population processes $Z_{a,1}^K(t)$ and $Z_{A,2}^K(t)$
 have small values, the processes $Z_{A,1}^K(t)$ and $Z_{a,2}^K(t)$ stay close to $\zeta_A$ and $\zeta_a$ respectively.

 Then, by bounding death rates, birth rates and migration rates of $(Z^K_{a,1}(t),t\geq 0)$ and $(Z^K_{A,2}(t),t\geq 0)$, we are able to compare the dynamics of these two processes with the ones of 
$$\left(\frac{\mathcal{N}_a(t)}{K},\frac{\mathcal{N}_A(t)}{K}, t\geq 0 \right),$$
 where $(\mathcal{N}_a(t),\mathcal{N}_A(t)) \in \N^{\{a,A\}}$ is a two-types branching process with types $a$ and $A$ and for which 
\begin{itemize}
 \item any $\alpha$-individual gives birth to a $\alpha$-individual at rate $b$,
 \item any $\alpha$-individual gives birth to a $\bar\alpha$-individual at rate $m_{\bar\alpha}$,
 \item any $\alpha$-individual dies at rate $b\beta_{\bar \alpha}+m_\alpha$.
\end{itemize}
The goal is thus to estimate the extinction time of such a sub-critical two types branching process. 
Let $M(t)$ be the mean matrix of the multitype process, that is,
$$
M(t)=
\begin{pmatrix}
 \E\Big[\E\Big[\mathcal{N}_a(t)\Big|(\mathcal{N}_a(0),\mathcal{N}_A(0))=(1,0)\Big]\Big] & & \E\Big[\E\Big[\mathcal{N}_A(t)\Big|(\mathcal{N}_a(0),\mathcal{N}_A(0))=(1,0)\Big]\Big]\\
& & \\
\E\Big[\E\Big[\mathcal{N}_a(t)\Big|(\mathcal{N}_a(0),\mathcal{N}_A(0))=(0,1)\Big]\Big] & & \E\Big[\E\Big[\mathcal{N}_A(t)\Big|(\mathcal{N}_a(0),\mathcal{N}_A(0))=(0,1)\Big]\Big]
\end{pmatrix},
$$
and let $G$ be the infinitesimal generator of the semigroup $\{M(t),t\geq 0\}$. From the book of~\cite{athreya1972branching} p.202, we deduce a formula of $G$ which is
$$
G=\begin{pmatrix}
   -b(\beta_A-1)-m_a & m_A \\
m_a & -b(\beta_a-1)-m_A
  \end{pmatrix}.
$$
Applying Theorem 3.1 of~\cite{heinzmann2009extinction}, we find that
\begin{equation}
\label{eq_probabitype}
\P\Big((\mathcal{N}_a(t),\mathcal{N}_A(t))=(0,0) \Big|(\mathcal{N}_a(0),\mathcal{N}_A(0))=(z_{a,1}^0 K,z_{A,2}^0 K)\Big)= (1-c_a e^{rt})^{z_{a,1}^0 K}(1-c_A e^{rt})^{z_{A,2}^0 K},
\end{equation}
where $c_a, c_A$ are two positive constants and $r$ is the largest eigenvalue of the matrix $G$. With a simple computation, we find that $r=-\omega(A,a)$. From~\eqref{eq_probabitype}, we deduce that the extinction time is of order $\omega(A,a)^{-1}\log K$ when $K$ tends to $+\infty$ by arguing as in step 2 of Proposition 4.1's proof of \cite{coron2016stochastic}. This concludes the proof of Lemma~\ref{lem_stoasym}.
\end{proof}

Finally, this gives all elements to induce Theorem~\ref{theo_asym}.\\

\textbf{Acknowledgements}:
The author would like to thank Pierre Collet for his help on the theory of dynamical systems.
This work  was partially funded by the Chair "Mod\'elisation Math\'ematique et Biodiversit\'e" of VEOLIA-Ecole Polytechnique-MNHN-F.X.\\

\bibliographystyle{plain}      
\bibliography{biblio_speciation}

\end{document}